\documentclass[10pt, conference]{IEEEtran}
\usepackage{amsthm}
\usepackage{amssymb}
\usepackage{graphicx}
\usepackage{amsmath}
\usepackage{newtxmath}
\usepackage[cal=cm]{mathalpha} 
\usepackage{tablefootnote} 
\usepackage{threeparttable} 
\usepackage{bm}
\newtheorem{theorem}{Theorem}{}
{}
{}
{}
{}
{}

\usepackage{algorithm}
\usepackage{algpseudocode}
\usepackage{graphics}
\usepackage{epsfig}
\usepackage{multirow}
\usepackage{caption}
\usepackage{array}
\usepackage{cite}
\usepackage{stfloats}
\usepackage{subfig} 
\usepackage{booktabs}
\usepackage{comment}
\usepackage{color}
\usepackage{tabu} 
\usepackage{stfloats}
\usepackage{mathtools}
\usepackage{stmaryrd} 
\makeatletter
\let\sum\relax
\let\prod\relax
\DeclareSymbolFont{largesymbols}{OMX}{cmex}{m}{n}
\DeclareMathSymbol{\sum}{\mathop}{largesymbols}{"50}
\DeclareMathSymbol{\prod}{\mathop}{largesymbols}{"51}
\usepackage[colorlinks,linkcolor=blue,anchorcolor=blue,citecolor=blue,bookmarks=true]{hyperref}
\usepackage[capitalize]{cleveref} 
\crefname{figure}{Fig.}{Figs.}   
\Crefname{figure}{Fig.}{Figs.}   

\DeclareMathOperator*{\diag}{\mathrm{diag}}

\DeclareMathOperator*{\tr}{\mathrm{tr}}

\DeclareMathOperator*{\var}{\mathrm{var}}

\begin{document}
	\title{{\Huge Movable Antenna Enhanced Wireless Sensing via Steering Vector Correlation and CRB Optimization}}
				\IEEEoverridecommandlockouts
	\vspace{-3em}
	\author{\IEEEauthorblockN{
			Chengzhi Ye\IEEEauthorrefmark{1}\IEEEauthorrefmark{2}, Ruoyu Zhang\IEEEauthorrefmark{2}, Shichao Wei\IEEEauthorrefmark{1}\IEEEauthorrefmark{2}, Wen Wu\IEEEauthorrefmark{2}, Byonghyo Shim\IEEEauthorrefmark{3},}
		\IEEEauthorblockA{\IEEEauthorrefmark{1}Qian Xuesen College, Nanjing University of Science and Technology, Nanjing 210094, China.
		}
		\IEEEauthorblockA{\IEEEauthorrefmark{2}Key Laboratory of Near-Range RF Sensing ICs and Microsystems (NJUST), Ministry of Education,\\ School of Electronic and Optical Engineering, Nanjing University of Science and Technology, Nanjing 210094, China.
		}
		\IEEEauthorblockA{\IEEEauthorrefmark{3}Department of Electrical and Computer Engineering, Seoul National University, Seoul 08826, South Korea.
		}
		Email: qxsycz8166@njust.edu.cn, ryzhang19@njust.edu.cn, weishichao@njust.edu.cn, wuwen@njust.edu.cn, bshim@snu.ac.kr. 
		\thanks{
			This work was supported in part by the National Natural Science Foundation of China under Grant 62571248, in part by the National Research Foundation of Korea (NRF) grant funded by the Korea government(MSIT) (2022M3C1A3099336), in part by the Undergraduate Research
			Training Program of Nanjing University of Science and Technology under
			Grant 202510288068 and Grant S125. (Corresponding author: Ruoyu Zhang).
		}
	}
	\vspace{-2em}
	
	
	\maketitle
	\vspace{-2em}


	%
	\vspace{-2em}
	\IEEEpeerreviewmaketitle
	\vspace{-2em}
	
	\begin{abstract}
		
		In this paper, we investigate the angle-of-arrival (AoA) estimation problem for wireless sensing systems equipped with movable antennas (MA). To achieve high estimation performance and accuracy, we formulate a joint optimization problem integrating the sidelobes of steering vector correlation (SVC) and the Cramér-Rao bound (CRB). We first mathematically transform the SVC and the CRB into tractable objective functions. Specifically, we introduce a proxy variable and apply a discrete grid search strategy to overcome the intractability of optimizing the SVC with unknown target angles. Concurrently, we derive a generalized lower bound for the CRB, which yields a scalar function of the MA positions. Guided by the transformed objective, we propose a successive convex approximation-based position optimization algorithm. The proposed algorithm handles the non-convex terms by employing first order Taylor expansions within a defined trust region, which allows the MA positions to be updated incrementally in each iteration. Simulation results demonstrate that the proposed algorithm achieves superior AoA estimation performance.
		
	\end{abstract}
	

\vspace{-0.5em}
	\section{Introduction}
	The forthcoming sixth-generation (6G) wireless networks are anticipated to support a massive proliferation of location-aware applications, including autonomous driving, robotic navigation, and immersive extended reality \cite{Saad6G}. To enable these advanced use cases, future networks must simultaneously satisfy stringent quality of service requirements for reliable data transmission and provide unprecedented high-precision environmental perception \cite{Jiang6g}. Driven by these dual demands, integrated sensing and communication (ISAC) has emerged as a transformative paradigm to jointly provide both services by sharing hardware and spectral resources \cite{ISACXUewen}. Consequently, wireless sensing is envisioned to become a primary service in 6G, encompassing the detection, localization, and extraction of physical information from surrounding targets \cite{DQLISAC2026}. 
	
	To achieve superior sensing performance with high angular resolution, we need to deploy large-scale antenna arrays, such as massive multiple-input multiple-output (MIMO) antenna array \cite{RuoyuUnified2024,ruoyu2025Large,ChengzhiTWCMIMOOFDM, ye2026RAsensing}. We note that the performance gains of such massive arrays are inevitably accompanied by prohibitive hardware costs, escalated energy consumption, and immense computational burdens \cite{UltramassiveMIMO}. As a cost-effective alternative, sparse antenna arrays have been extensively investigated. The essence of this approach is to attain comparable spatial resolution by enlarging the inter-antenna spacing with significantly fewer elements \cite{SparseMIMOsensing}. Fundamentally, both conventional massive arrays and sparse arrays rely on fixed-position antennas (FPAs) \cite{CanSparse}. So the FPA-based architectures inherently lack the flexibility to dynamically switch between optimal array geometries. 
	
	To tackle the challenges, movable antenna (MA) systems have garnered considerable interest in both communication and sensing domains, primarily owing to their capability to dynamically optimize channel conditions via intelligent antenna placement \cite{11218873,ye2026selfcalibrationdoaestimationmovable,maCompressedSensingBased2023,BNIngMovableAntennaEnhanced2025,zhangChannelEstimationMovableantenna2024,Chengzhi2025,ChengzhiIOTJ20262D}. In contrast to conventional FPAs, MA arrays facilitate flexible antenna movement, thereby enabling the comprehensive exploitation of spatial degrees of freedom (DoFs) within a confined deployment region. 
	Recent literature has validated the superior wireless sensing capabilities of MA architectures, revealing that strategic antenna positioning substantially boosts overall sensing accuracy \cite{maMovableAntennaEnhanced2024,Shaoxiaodan6dsensing, 11456856}. 
	For instance, \cite{maMovableAntennaEnhanced2024} primarily focus on optimizing MA positions by minimizing the Cramér-Rao bound (CRB) for single-target scenarios. However, the resulting array configurations cannot be directly extrapolated to multi-target environments. Such single-target optimization inadvertently aggravates the steering vector correlation (SVC) associated with different spatial directions, causing severe sidelobes with pronounced peak levels in the spatial spectrum, which act as strong interference and significantly degrade the estimation accuracy of other concurrent target sources. 
	
	Motivated by these critical limitations, this paper investigates a  MA positions optimization framework for wireless sensing. Specifically, we formulate a joint optimization problem aiming to simultaneously minimize the sidelobes of the SVC and the CRB. To enhance mathematical tractability, on one hand, we transform the SVC into a tractable objective function by introducing a proxy variable and employing a discrete grid search strategy across the spatial domain. On the other, we derive a generalized lower bound for the CRB, which yields a scalar function expression of the MA positions.
	Building upon this analytical insight, we propose a successive convex approximation (SCA)-based iterative algorithm to solve the joint optimization problem iteratively. The proposed algorithm handles the non-convex terms by employing first order Taylor expansions, which allows the MA positions to be updated incrementally in each iteration. This iterative framework ensures that the optimized array configuration achieves a converged balance between the minimization of the sidelobes of the SVC and the lower bound of the CRB for the MA positions. Simulation results demonstrate that the proposed MA positions effectively suppress sidelobes and achieves superior estimation performance compared to conventional benchmark schemes.

		Notations: Symbols $a$, $\bm{a}$, $\bm{A}$, $\mathcal{A}$ denote a scalar, a vector, a matrix, and a set, respectively. The trace, inverse and the real part of a matrix are represented by $\tr\{\cdot\}$, $(\cdot)^{-1}$, and $\mathfrak{R}\{\cdot \}$.
		The linear space spanned by matrix $\bm{A}$ is denoted by $\mathrm{span}{ (\bm{A}) }$.
		The $K\times K$ dimensional identity matrix is expressed by $\bm{I}_{K}$. The $N \times 1$ dimensional vector with all the elements equal to 1 is given by $\bm{1}_{N}$.

	\section{System Model and Problem Formulation}
	\begin{figure}[t]
		\vspace{-1em}
		\centering
		\includegraphics[width=0.75\linewidth]{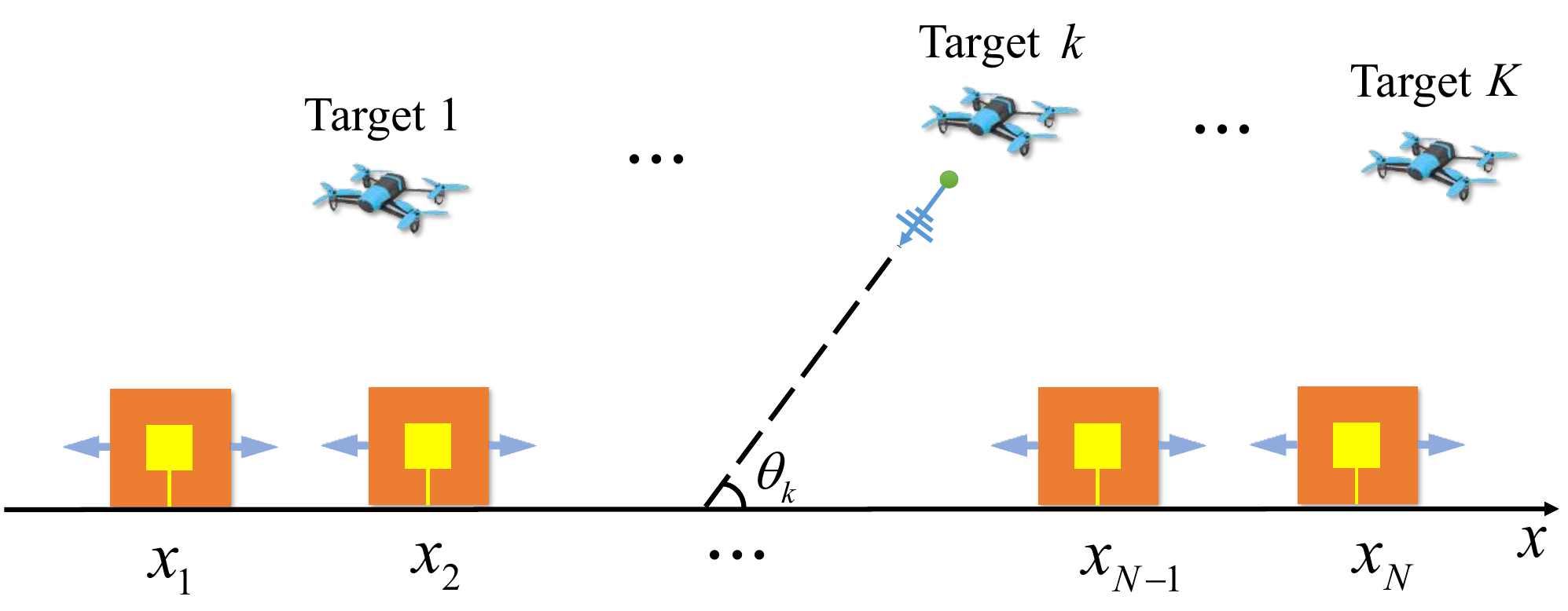}
		\caption{The MA array for targets sensing.}
		\label{MAUAV}
		\vspace{-1em}
	\end{figure}
	In this section, we consider a wireless sensing model with $N$ MAs equipped with driver components to estimate the AoAs of $K$ targets in \Cref{MAUAV}. The MA positions can be adjusted flexibly within the given line region of length $A$. The APV of all $N$ MAs is denoted by $\bm{x}=[x_1,x_2,\ldots,x_N]^T$ with $x_n\in [0,A]$.
	Since the distance between the target and receiver is typically much larger than the size of the region for antenna movement, we consider the far-field channel model from the target to receiver \cite{maMovableAntennaEnhanced2024}. As shown in \Cref{MAUAV}, the AoA of the $k$-th target-receiver LoS path is denoted by $\theta_k\in \mathcal{A}$, where $\mathcal{A} = \left\{ \vartheta \mid \vartheta_{\min} \le \vartheta \le \vartheta_{\max} \right\}$ represents the set of all possible AoAs of the targets. The array steering vector of the MA array for the $k$-th target can thus be expressed as a function of $\bm{x}$ and $ \theta_k$
	\begin{align}
		\bm{a}(\bm{x}, \theta_k)=\big[e^{j\kappa x_1\cos \theta_k},\ldots,e^{j\kappa x_N\cos \theta_k}\big]^T\in \mathbb{C}^{N\times 1},
	\end{align}
	where $\kappa = \frac{2\pi}{\lambda}$ with $\lambda$ denoting the wavelength.
	The $t$-th snapshot of the received signal $\bm{y}$ can be expressed as
	\begin{align}
		\bm{y}(t)=\sum\limits_{k=1}^{K}{\bm{ a}({\bm{x},\theta_k}){{s}_{k}}(t)}+\bm{z}(t), \label{sys1}
	\end{align}
	where ${{s}_{k}}(t)$ is the $k$-th signal and $\bm{z}(t)$ is the additive noise which obeys zero-mean white Gaussian distribution, i.e., $\bm{z}(t)\sim\mathcal{CN}(\bm{0}_{N\times1},\sigma_z^2\bm{I}_{{N}})$. Considering $T$ snapshots, the received signal in \eqref{sys1} can be expressed as
	\begin{align}
		\bm{Y}=\bm{ AS}+\bm{Z}, \label{sys2}
	\end{align}
	where $\bm{Y}={{[ {{\bm{y}}(1)},{{\bm{y}}(2)}, \ldots ,{{\bm{y}}(T)} ]}} \in \mathbb{C}^{N \times T}$, $\bm{S}={{[ {{\bm{s}}(1)},{{\bm{s}}(2)},\ldots ,{{\bm{s}}(T)} ]}\in \mathbb{C}^{K \times T}}$, $\bm{A}=[ \bm{a}(\bm{x},{{\theta }_{1}}),\bm{a}(\bm{x},{{\theta }_{2}}),\ldots ,\bm{a}(\bm{x},{{\theta }_{K}}) ]\in \mathbb{C}^{N \times K}$, and $\bm{Z}=[\bm{z}(1),\bm{z}(2),\ldots,\bm{z}(T)]\in \mathbb{C}^{N \times T}$. We adopt the multiple signal classification (MUSIC) algorithm for estimating the AoAs of the targets. 
	
	First, the covariance matrix of $\bm{Y}$ can be written as
	\begin{align}
		\bm{R}=\frac{1}{T}\bm{YY}^H=\frac{1}{T}\bm{AS}\bm{S}^H\bm{A}^H+{\sigma_z^2}\bm{I}_{N} \in \mathbb{C}^{N \times N},
	\end{align}
	where $\sigma_z^2$ denotes the noise power. Using the standard MUSIC algorithm, we can obtain the singular value decomposition of $\bm{R}$, i.e.,
	\begin{align}
		\bm{R} &= \bm{U}_{s}\bm{\varSigma}_{s}\bm{U}_{s}^H + \bm{U}_{z}\bm{\varSigma}_{z}\bm{U}_{z}^H,
	\end{align}
	where $\bm{U}_s\in \mathbb{C}^{N\times K}$ and $\bm{U}_z\in \mathbb{C}^{N\times (N-K)}$ are the singular vectors of the signal and noise subspaces, respectively. $\bm{\varSigma}_{s}\in \mathbb{R}^{K\times K}$ and $\bm{\varSigma}_{z}\in \mathbb{R}^{(N-K)\times (N-K)}$ are diagonal matrices with the diagonal elements representing the singular values of the signal and noise subspaces, respectively. Then, the spatial spectrum $P(\vartheta)$ is computed as
	\begin{align}\label{Pthtea}
		P(\vartheta)= {1}/{\bm{a}^H(\bm{x},\vartheta)\bm{{U}}_z\bm{U}_z^H\bm{a}(\bm{x},\vartheta)}.
	\end{align}
	Based on the spatial spectrum in \eqref{Pthtea}, the estimated AoAs $\bm{\hat{\theta}}=[\hat{\theta}_1,\ldots,\hat{\theta}_K]^T$ can be obtained corresponding to $K$ principal maxima. 
	To evaluate the mitigation of sidelobes, we employ the SVC gauging the suppression of sidelobes relative to the main lobe as a metric. The SVC between the AoA of the $k$-th target and $\vartheta$ can be given by
	\begin{align}
		\mathrm{SVC}(\bm{x}, \theta_k)&=\frac{1}{N^2}\big|\bm{a}^{H}(\bm{x}, \theta_k)\bm{a}(\bm{x},\vartheta)\big|^2\\
		&=\frac{1}{N^2}\Big|\!\sum_{n=1}^{N}{e^{\kappa x_n(\cos \theta_k-\cos \vartheta)}}\Big|^2\label{SVC}.\nonumber
	\end{align}
	Another well-known metric of wireless sensing performance is the CRB, which can be given by \cite{9573279}
	\begin{align}
		&\mathrm{CRB}(\bm{x},{\theta}_k) \\
		&=\frac{\sigma_z^2}{2}\sum_{t=1}^T\Big[\mathfrak{R}\big\{\bm{\tilde{S}}^H(t)\bm{B}^H(\bm{x},\bm{\theta})\bm{P}_{\bm{A}}^\perp(\bm{x},\bm{\theta})\bm{B}(\bm{x},\bm{\theta})\bm{\tilde{S}}(t)\big\}\Big]_{k,k}^{-1}\nonumber,
	\end{align}
	where $\bm{\tilde{S}}(t)=\diag[s_1(t),\ldots,s_K(t)]$, $\bm{P}_{\bm{A}}^\perp (\bm{x},\bm{\theta})=\bm{I}_{N}-\bm{A}\big(\bm{A}^H\bm{A}\big)^{-1}\bm{A}^H$, and $\bm{B}(\bm{x},\bm{\theta})=[\bm{\dot{a}}_{1},\ldots,\bm{\dot{a}}_{K}]$ with $\bm{\dot{a}}_{k}=\frac{\partial \bm{a}(\bm{x}, \theta_k)}{\partial  \theta_k} = -j\frac{2\pi}{\lambda}\sin\theta_k \diag(\bm{x})\bm{a}_k$.
	
	Our goal is to minimize the SVC corresponding to AoAs distinct from the targets while simultaneously minimizing the CRB via dynamically adjusting the MA positions. The optimization problem can be formulated as
	{\begin{subequations}\label{P1}
			\begin{align}
				(\text{P1}) \enspace \min_{\bm{x}} \enspace & \xi\mathrm{SVC}(\bm{x}, \theta_k) + (1-\xi) \mathrm{CRB}(\bm{x}, {\theta}_k) \label{P1rhoop} \\
				\text{s.t.} \enspace  
				& x_n \in [0, A], \quad n=1,\ldots, N \label{P1const1} \\
				& x_n - x_{n-1} \ge d, \quad n=2,\ldots, N, \label{P1const2}
			\end{align}
		\end{subequations}
		where $0\le \xi \le 1$ is a weighting coefficient for the SVC and the CRB. We note that the problem (\hyperref[P1]{P1}) is non-convex and also very difficult due to the complexity of the objective function. To handle this problem, we propose a joint SVC and CRB optimization algorithm to achieve the optimal tradeoff between the SVC and the CRB.
	\section{Proposed Joint SVC and CRB Optimization Algorithm}
	In this section, we first reformulate the optimization problem (\hyperref[P1]{P1}) and subsequently propose a joint SVC and CRB optimization algorithm based optimization algorithm. Specifically, variable substitution is performed for the SVC to effectively eliminate the a priori unknown AoA from the mathematical expressions. Regarding the CRB, the lower bound of the CRB for the AoAs is derived and position dependent closed form functions are extracted to facilitate the optimization process. Finally, the objective function and constraints are approximated via first order Taylor expansion to ensure convexity, which enables the optimization of $\bm{x}$ through successive iterations.
	\subsection{Problem Transformation}
	\begin{figure}[t]
		\vspace{-1em}
		\centering
		\includegraphics[width=0.75\linewidth]{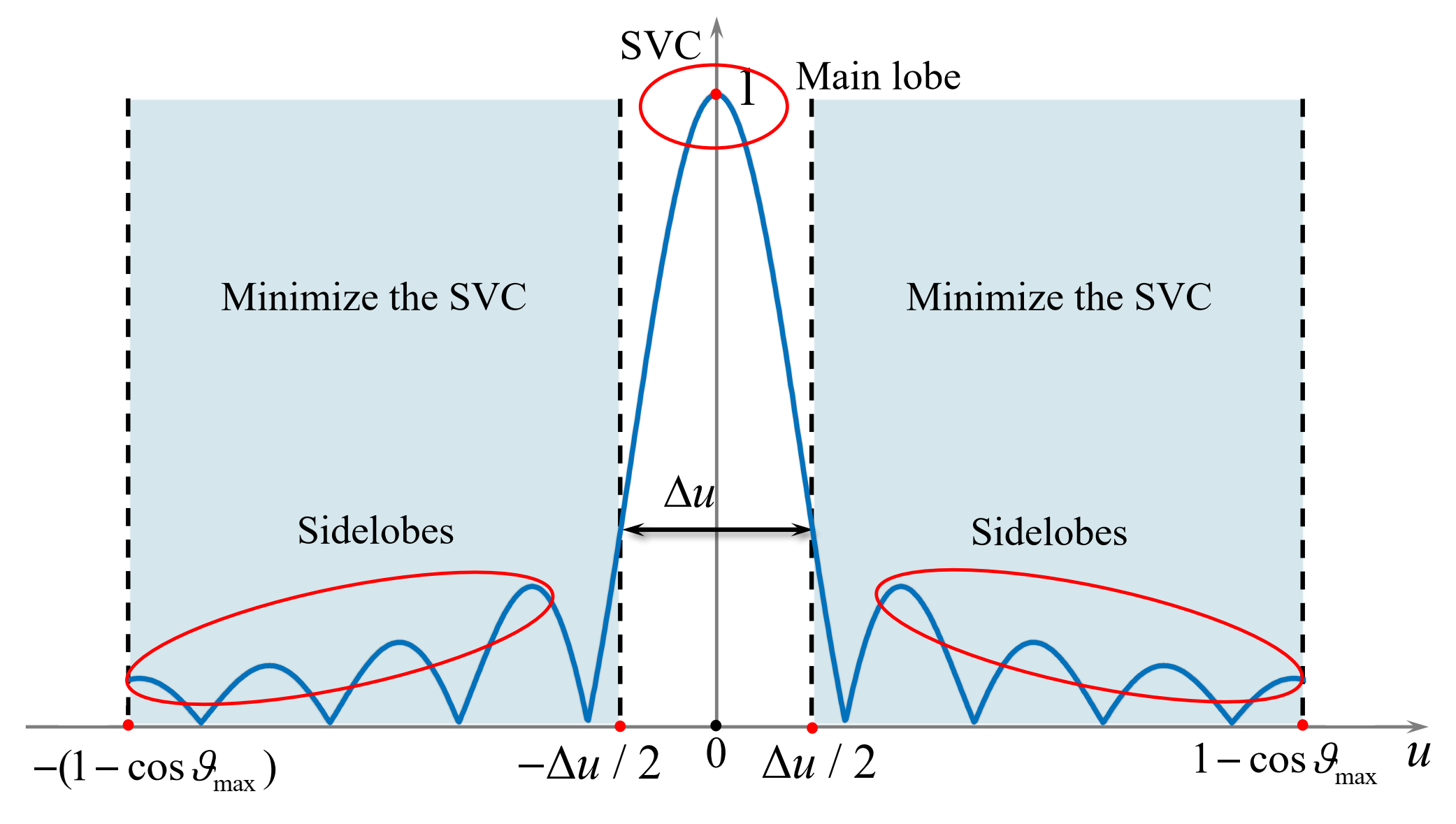}
		\caption{Diagram of $\mathrm{SVC}(\bm{x}, u)$.}
		\label{Fig}
		\vspace{-1em}
	\end{figure}
	
	To improve the tractability of the optimization problem, we first reformulate the spatial metric over the cosine difference domain. Since $\theta_k$ and $\vartheta$ belong to $\mathcal{A}$, the term $(\cos \theta_k - \cos \vartheta)$ for $k = 1, \ldots, K$ spans the interval $[-(1-\cos \vartheta_{\max}), 1-\cos \vartheta_{\max}]$. Although the specific angles $\theta_k$ are unknown a priori, we can define a variable $u$ within this interval to equivalently represent all potential values of the cosine difference. Consequently, the expression of $\mathrm{SVC}(\bm{x}, \theta_k)$ for $k = 1, \ldots, K$ can be transformed to $\mathrm{SVC}(\bm{x}, u)$, which can be given by \eqref{SVChua},
	\begin{figure*}[t]
		\vspace{-1em}
		\begin{align}\label{SVChua}
		\!\!\mathrm{SVC}(\bm{x}, u) &\!=\! \frac{1}{N^2} \big| \sum_{n=1}^{N} e^{j \kappa x_n u} \big|^2\!=\!\frac{1}{N^2} \big(\sum_{n=1}^{N} e^{j \kappa x_n u}\big) \big(\sum_{n=1}^{N} e^{-j \kappa x_n u}\big)\!=\!\frac{1}{N^2}\sum_{n=1}^{N}\sum_{m=1}^{N}e^{j\kappa u(x_n - x_m)}
		\!=\! \frac{1}{N^2}\sum_{n=1}^{N}\sum_{m=1}^{N}\cos\big[{\kappa u(x_n - x_m)}\big],
	\end{align}
	\vspace{-0.5em}
		\hrule
		\vspace{-1em}
	\end{figure*}
	where the fourth equation holds true by the utilization of the Euler's formula.
	The curve of the $\mathrm{SVC}(\bm{x}, u)$ is illustrated in \Cref{Fig}. It can be observed that $\mathrm{SVC}(\bm{x}, u)$ deviates from an ideal impulse function. Specifically, it reaches its peak value of exactly $1$ at $u=0$ and inevitably remains close to $1$ within the vicinity of the main lobe. Furthermore, $\mathrm{SVC}(\bm{x}, u)$ exhibits numerous sidelobes. To effectively suppress these sidelobes while preserving the main lobe characteristics, we define a main lobe width denoted by $\Delta u$. We purposely avoid optimizing the SVC within the interval $[-\Delta u/2, \Delta u/2]$ to protect the main lobe integrity and exclusively minimize the SVC outside this region to suppress the sidelobes. Recognizing that $\mathrm{SVC}(\bm{x}, u)$ is an even function with respect to $u$, we only need to consider the optimization over the positive half interval $\mathcal{U} = [\Delta u/2, 1-\cos \vartheta_{\max}]$.
	
	Next, we proceed to simplifying the CRB. It should be noted that $\mathrm{CRB}(\bm{x}, \theta_k)$ is a function of $\bm{x}$. However, due to the presence of numerous terms dependent of $\bm{x}$, the CRB cannot be further formulated into a tractable form conducive to solving the optimization problem. In the following, we present Theorem \ref{theorem1} to derive a lower bound on the CRB by explicitly separating the terms related to $\bm{x}$.
	\begin{theorem}\label{theorem1}
		For AoA estimation, the lower bound of $\mathrm{CRB}(\bm{x},{\theta_k})$ can be given by 
		\begin{align}
			\mathrm{CRB}(\bm{x}, {\theta}_k)\geq  {\frac{\sigma_z^2 \lambda^2}{8\pi^2 N\sin ^2 \theta_k TV{(\bm{x})}}},
		\end{align}
		where $V(\bm{x})=\frac{1}{N}\sum_{n=1}^{N}{x^2_n}-\big(\frac{1}{N}\sum_{n=1}^{N}{x_n}\big)^2$.
	\end{theorem}
	\begin{proof}
		The Fisher information matrix (FIM) of $\bm{\theta}$ is given by \cite{9573279}
		\begin{align}
			\bm{F}=\frac{2}{\sigma_z^2}\sum_{t=1}^T\tr\left\{\mathfrak{R}\Big\{\bm{\tilde{S}}^H(t)\bm{B}^H\bm{P}_{\bm{A}}^\perp\bm{B}\bm{\tilde{S}}(t)\Big\}\right\}.
		\end{align}
		Based on the properties of the inverse matrix, we have
		\begin{align}
			\mathrm{CRB}(\bm{x}, \theta_k)={\big[\bm{F}^{-1}\big]_{k,k}} \geq  {{F}_k^{-1}},
		\end{align}
		where $F_k$ is the $k$-th main diagonal element corresponds to the parameter $\cos \theta_k$ of the $k$-th target, which can be explicitly expressed as \eqref{FFk},
		\begin{figure*}
			\vspace{-1em}
			\begin{align}\label{FFk}
			F_k &= \Big[ \frac{2}{\sigma_z^2}\sum_{t=1}^T\mathfrak{R}\Big\{\bm{\tilde{S}}^H(t)\bm{B}^H\bm{P}_{\bm{A}}^\perp\bm{B}\bm{\tilde{S}}(t)\Big\} \Big]_{k,k}= \frac{2}{\sigma_z^2}\sum_{t=1}^T \mathfrak{R}\Big\{ \big( \bm{B}\bm{\tilde{S}}(t) \big)_{:,k}^H \bm{P}_{\bm{A}}^\perp \big( \bm{B}\bm{\tilde{S}}(t) \big)_{:,k} \Big\} \nonumber\\
			&= \frac{2}{\sigma_z^2}\sum_{t=1}^T \mathfrak{R}\Big\{ \big( s_k(t)\bm{\dot{a}}_k \big)^H \bm{P}_{\bm{A}}^\perp \big( s_k(t)\bm{\dot{a}}_k \big) \Big\}= \frac{2}{\sigma_z^2}\sum_{t=1}^T \mathfrak{R}\Big\{ |s_k(t)|^2 \bm{\dot{a}}_k^H \bm{P}_{\bm{A}}^\perp \bm{\dot{a}}_k \Big\}= \frac{2}{\sigma_z^2} T \mathfrak{R} \Big\{ \bm{\dot{a}}_k^H \bm{P}_{\bm{A}}^{\perp} \bm{\dot{a}}_k \Big\},
		\end{align}
		\vspace{-0.5em}
			\hrule
			\vspace{-1em}
		\end{figure*}
		where we define the single-target projection matrix as $\bm{P}_{\bm{a}_k}^{\perp} = \bm{I}_{N} - \frac{1}{N}\bm{a}_k \bm{a}_k^H$. Given that $\mathrm{span}(\bm{a}_k) \subseteq \mathrm{span}(\bm{A})$, their orthogonal complements satisfy the reverse inclusion, i.e., $\mathrm{span}(\bm{A})^{\perp} \subseteq \mathrm{span}(\bm{a}_k)^{\perp}$, which directly establishes the positive semi-definite relation $\bm{P}_{\bm{a}_k}^{\perp} \succeq \bm{P}_{\bm{A}}^{\perp}$. Substituting it into the expression of ${F}_k$, we obtain the upper-bound FIM element of ${\bar{F}}_k$, i.e.,
		\begin{align}
			{F}_k \leq \frac{2}{\sigma_z^2} T \mathfrak{R} \left\{ \bm{\dot{a}}_k^H \bm{P}_{\bm{a}_k}^{\perp} \bm{\dot{a}}_k \right\} \triangleq {\bar{F}}_k
			\label{barFk}.
		\end{align}
		By submitting the expression of $\bm{\dot{a}}_k$, we obtain
		\begin{align}
			\!\!\!\!\mathrm{CRB}(\bm{x}, \theta_k)&\ge{{\bar{F}}_k^{-1}} = 
			{\frac{\sigma_z^2 \lambda^2}{8\pi^2 N \sin^2\theta_k TV (\bm{x})}},
		\end{align}
		which establishes Theorem \ref{theorem1}.
		\vspace{-0.5em}
	\end{proof}
	Theorem \ref{theorem1} derives a lower bound on $\mathrm{CRB}(\bm{x})$, which facilitates the isolation of the $\bm{x}$-dependent scalar function $1/V(\bm{x})$. The equality holds if and only if $K=1$, in which case the bound reduces to the single-target CRB  \cite{maMovableAntennaEnhanced2024}. For multi-target scenarios, the presence of multiple targets only scales the coefficient of $1/V(\bm{x})$ without altering the functional form of $V(\bm{x})$ itself. Thus, the derived lower bound serves as a generalized form of the one in \cite{maMovableAntennaEnhanced2024}. Moreover, this bound reveals that the CRB is inversely proportional to $V(\bm{x})$, implying that a larger $V(\bm{x})$ leads to a smaller CRB.
	
	A fundamental prerequisite in multi-objective formulations is ensuring commensurate scaling among distinct metrics to prevent scale-induced domination in the combined cost function. Since the auxiliary variable $\tau$ bounds the spatial correlation coefficient, its value is inherently restricted to the interval $[0, 1]$, thereby obviating the need for further normalization. Conversely, the dynamic range of $V(\bm{x})$ is intrinsically tied to the physical array aperture and system parameters, meaning its scale can arbitrarily deviate from $[0, 1]$. Consequently, it is imperative to normalize $V(\bm{x})$ to a standard $[0, 1]$ interval prior to applying the linear weights. The lower bound and upper bound of $V(\bm{x})$ can be given by
	\begin{align}
		V_{\text{low}}=\frac{(N^2-1)d^2}{12}\le V(\bm{x}) < \frac{A^2}{4}=V_{\text{up}},
	\end{align}
	where the lower bound $V_{\text{low}}$ is strictly achieved when all antennas are most densely packed into a uniform linear array (ULA) configuration, constrained by the minimum spacing $d$ and the absolute upper bound $V_{\text{up}}$ is governed by Popoviciu's inequality on the variable $\bm{x}$ in $[0, A]$. Having established the explicit dynamic range $[V_{\text{low}}, V_{\text{up}})$, we can rigorously map the physically dependent aperture variance into a dimensionless, normalized metric $\tilde{V}(\bm{x}) \in [0, 1)$. This standard min-max scaling is formulated as follows
	\begin{align}
		\tilde{V}(\bm{x}) = \frac{V(\bm{x}) - V_{\text{low}}}{V_{\text{up}} - V_{\text{low}}}.
	\end{align}
	The original optimization problem is formulated as
	\begin{subequations}\label{P2} 
		\begin{align}
			(\text{P2}) \enspace \min_{\bm{x}} \enspace & \xi\mathrm{SVC} (\bm{x}, u) -(1-\xi) \tilde{V}(\bm{x}) \label{rhoopa} \\
			\text{s.t.} \enspace 
			& \tilde{V}(\bm{x}) \ge \tilde{V}_{\text{th}} \label{const_vtha} \\
			& x_n \in [0, A], \quad n=1,\ldots, N \label{const1a} \\
			& x_n - x_{n-1} \ge d, \quad n=2,\ldots, N, \label{const2a}
		\end{align}
	\end{subequations}
	where the constraint \eqref{const_vtha} serves as a requirement for the estimation performance and $V_{\text{th}}$ is a predefined variance threshold. 
	To solve this problem efficiently, we address the positive half interval associated with the variable $u$ by discretizing $\mathcal{U}$ into a finite set $\{u_l\}_{l=1}^L$, where $u_l \in [\Delta u/2, 1-\cos \vartheta_{\max}]$ can be given by
	\begin{align}
		u_l = \frac{\Delta u}{2} + \frac{l - 1}{L - 1} \Big( 1 - \cos \vartheta_{\max} - \frac{\Delta u}{2} \Big), l = 1, 2, \dots, L,
	\end{align}
where $L$ denotes the number of the discrete points of $\mathcal{U}$. We aim to optimize the worst case performance over this finite set. Specifically, we introduce an auxiliary variable $\tau$ to serve as the upper bound for the $\mathrm{SVC}$ term across all discrete points in $\mathcal{U}$. By substituting the worst case term in the objective function with $\tau$ and adding a new set of inequality constraints, we can transform (\hyperref[P2]{P2}) into a more tractable form
	\begin{subequations}\label{P3} 
		\begin{align}
			(\text{P3}) \enspace \min_{\bm{x}} \enspace & \xi\tau -(1-\xi) \tilde{V}(\bm{x}) \label{rhoop} \\
			\text{s.t.} \enspace  & \mathrm{SVC}(\bm{x}, u_l) \le \tau, \enspace l=1,\ldots,L \label{const_tau} \\
			& \tilde{V}(\bm{x}) \ge \tilde{V}_{\text{th}} \label{const_vth} \\
			& x_n \in [0, A], \quad n=1,\ldots, N \label{const1} \\
			& x_n - x_{n-1} \ge d, \quad n=2,\ldots, N , \label{const2}
		\end{align}
	\end{subequations}
	where the constraint \eqref{const_tau} ensures that the maximum of $\mathrm{SVC}(\bm{x}, u_l)$ over the considered grid points is strictly upper bounded by $\tau$. Constraints \eqref{const1} and \eqref{const2} represent the physical boundaries and spacing limits of MAs. Note that the problem (\hyperref[P3]{P3}) remains nonconvex. In the following we shall introduce an SCA-based algorithm to optimize the positions of MAs.
	
	\vspace{-0.5em}
	\subsection{Proposed SCA-Based Algorithm}
	
	It can be noted that the objective function \eqref{rhoop} and the constraints \eqref{const_tau} and \eqref{const_vth} are highly nonconvex, which makes the problem (\hyperref[P3]{P3}) mathematically intractable to solve directly. A major mathematical hurdle arises in the mutual coherence constraint \eqref{const_tau}. 
	To circumvent this issue, we can construct a surrogate function to locally approximate $\mathrm{SVC}(\bm{x}, u_l)$. To be specific, we first introduce a property that the second order Taylor expansion of $\cos (\alpha)$ at $\alpha_0$ is
	\begin{align}
		\cos(\alpha_0) \!-\! \sin (\alpha_0) (\alpha-\alpha_0) \!-\! \frac{1}{2}\cos(\alpha_0)(\alpha-\alpha_0)^2.
	\end{align}
	Since $\cos(\alpha_0) \ge -1$, a surrogate function for $\cos (\alpha)$ can be given by
	\begin{align}\label{cosa}
		\!\!\!\!\!\cos (\alpha) \!\le\! \cos (\alpha;\! \alpha_0) \!\triangleq\! \cos(\alpha_0) \!-\! \sin (\alpha_0) (\alpha\!-\!\alpha_0) \!+\! \frac{(\alpha\!-\!\alpha_0)^2}{2}.\!\!
	\end{align}
Then, for $\mathrm{SVC}(\bm{x}, u_l)$, we define $g_l(x_n, x_m) = \kappa u_l(x_n - x_m)$, its expression can be further given by
	\begin{align}\label{SVCcos}
		\mathrm{SVC}(\bm{x}, u_l) = \frac{1}{N^2} \sum_{n=1}^{N}\sum_{m=1}^{N} \cos [g_l (x_n, x_m)].
	\end{align}
	 Based on \eqref{cosa}, by replacing $\alpha$ with $g_l(x_n, x_m)$ and $\alpha_0$ with $g_l(x_n^i, x_m^i)$ for \eqref{SVCcos}, we can construct a surrogate function for $\mathrm{SVC}(\bm{x}, u_l)$ at the $i$-th iteration
	\begin{align}\label{SVCover}
		\mathrm{SVC}(\bm{x}, u_l) &\le \frac{1}{N^2}\sum_{n=1}^{N}\sum_{m=1}^{N}\cos \big[g_l(x_n \!-\! x_m); g_l (x_n^{i} \!-\!x_m^{i})\big] \nonumber\\ &\triangleq \overline{\mathrm{SVC}}(\bm{x}, u_l; \bm{x}^{i})
		 = \frac{1}{2}\bm{x}^T \bm{G}_l \bm{x} +  \bm{\beta}_l^T \bm{x} +\gamma_l,
	\end{align}
	where $\bm{G}_l \in \mathbb{R}^{N \times N}$, $\bm{\beta}_l \in \mathbb{R}^{N \times 1}$, and $\gamma_l \in \mathbb{R}$ are respectively given by
	\begin{align}
		\bm{G}_l &=\frac{2\kappa^2 u_l^2}{N^2}(N\bm{I}_N - \bm{1}_{N}\bm{1}_{N}^T),\\
		\![\bm{\beta}_l]_n &\!=- \frac{2}{N^2} \sum_{m=1}^{N}\big\{\kappa u_l \sin \big[g_l(x_n^{i},x_m^{i}) \big] \!+\! \kappa^2 u_l^2(x_n^{i} -x_m^{i})\big\},\!\!\\
		\gamma_{l}&=\frac{1}{N^2}\sum_{n=1}^N\sum_{m=1}^N\Big\{\cos\big[g_l(x_n^i, x_m^i)\big]\\&\enspace +g_l(x_n^i, x_m^i)\sin\big[g_l(x_n^i, x_m^i)\big]+\frac{1}{2}g_l^2(x_n^i,x_m^i)\Big\}.\nonumber
	\end{align}
	Note that \eqref{SVCover} is a standard quadratic form and can be solved by the convex optimization tools.
	
	Due to $\tilde{V} (\bm{x})$ is convex with respect to $\bm{x}$, we can perform the first order Taylor expansion of the $i$-th iteration at $\bm{x}^{i}$, i.e., 
	\begin{align} 
		\tilde{V}(\bm{x}) \ge \underline{\tilde{V}}(\bm{x}; \bm{x}^{i}) \triangleq\! \tilde{V}(\bm{x}^{i}) + \nabla \tilde{V}(\bm{x}^{i})^T (\bm{x}    - \bm{x}^{i})\label{V2Tay},
	\end{align}
	where $\nabla \tilde{V}(\bm{x}^{i}) = \frac{2(\bm{x}^{i} - \bar{x}^{i}\bm{1}_{N\times 1})}{N(V_{\text{up}}-V_{\text{low}})}$ with $\bar{x}^{i}=\frac{1}{N}\sum_{n=1}^{N}{x_{n}^{i}}$. 
	Thus, at the $i$-th SCA iteration, the original non-convex optimization problem (\hyperref[P3]{P3}) can be simplified as the following convex problem
	\begin{subequations}\label{P4}
		\begin{align}
			(\text{P4}) \enspace \min_{\bm{x}} \enspace & \xi\tau -(1-\xi) \underline{\tilde{V}}(\bm{x}; \bm{x}^{i}) \label{P2_obj} \\
			\text{s.t.} \enspace  & \overline{\mathrm{SVC}}(\bm{x}, u_l; \bm{x}^{i}) \le \tau, \enspace l=1,\ldots,L \label{P2_const_tau} \\
			& \underline{\tilde{V}}(\bm{x}; \bm{x}^{i}) \ge \tilde{V}_{\text{th}} \label{P2_const_vth} \\
			& x_n  \in [0, A], \quad n=1,\ldots, N \label{P2_const1} \\
			& x_n  - x_{n-1} \ge d, \quad n=2,\ldots, N. \label{P2_const2}
		\end{align}
	\end{subequations}
	For the problem (\hyperref[P4]{P4}), the objective function \eqref{P2_obj} and constraints \eqref{P2_const_vth}-\eqref{P2_const2} are linear, and \eqref{P2_const_tau} is a standard second order cone constraint. Therefore, (\hyperref[P4]{P4}) can be efficiently optimally solved by MATLAB CVX toolbox \cite{mosek2019}. The proposed algorithm is summarized in Algorithm \ref{Algorithm 1} and continues to refine the MA positions as long as the iteration index $i$ is no greater than the maximum limit $I$, or the relative residuals of the SVC and $\tilde{V}(\bm{x})$ remain above the tolerance $\delta$.
	
	\begin{algorithm}[t]
		\renewcommand{\algorithmicrequire}{\textbf{Input:}}
		\renewcommand{\algorithmicensure}{\textbf{Output:}}
		\caption{Proposed SCA-Based Optimization Algorithm}
		\label{Algorithm 1}
		\begin{algorithmic}[1]
			\Require 
			The number of MAs $N$, the movable region $A$, the minimum spacing $d$, the discretized spatial frequency set $\mathcal{U} = \{{u}_l\}_{l = 1}^{L}$, the maximum iteration $I$, and the parameters $\{\tilde{V}_{\text{th}}, \xi, \delta\}$.
			\State \textbf{Initialization:} Set iteration index $i = 1$, and the initial MA positions $\bm{x}^{0}$. 
			\While {not converged and $i \leq I$}
			\State Construct the surrogate functions \eqref{SVCover} and \eqref{V2Tay};
			\State Solve the standard convex problem (\hyperref[P4]{P4});
			\State $i \leftarrow i+1$;
			\EndWhile
			\Ensure Return the optimized MA positions $\bm{x}^\star$.
		\end{algorithmic}
	\end{algorithm}
	
	\section{Simulation Results}
	\label{Sec:Simulation}
	In this section, we evaluate the performance of the proposed joint SVC and CRB optimization algorithm through numerical simulations. The movable region is set to $A = 12\lambda$, the antenna minimum space is set to $\lambda/2$, the number of MAs is set to $N = 12$, and the number of snapshots is $T=10$, and the AoA is $\bm{\theta}=[45.1^{\circ},60.5^{\circ},120.9^{\circ}]^T$, $ \vartheta_{\min}=30^{\circ}$, and $\vartheta_{\max}=150^{\circ}$. The signal-to-noise ratio (SNR) is given by $\mathrm{SNR} = 10\lg{\sigma_k^2}/{\sigma_z^2}$. The considered benchmark schemes for setting MA positions are listed as follows: 
	\begin{itemize}
		\item \textbf{ULA}: Applies the conventional ULA with the antenna minimum space $d$, i.e., $\bm{x}_{\text{ULA}}=[0,d,\ldots,(N-1)d]^T$.
		\item \textbf{USA}: Applies the conventional uniform sparse array (USA) with the antenna space $d_{\text{USA}}=A/(N-1)$, i.e., $\bm{x}_{\text{USA}}=[0,d_{\text{USA}},\ldots,(N-1)d_{\text{USA}}]^T$.
		\item \textbf{CMA}: Applies the conventional MA (CMA) positions $\bm{x}_{\mathrm{CMA}}$ in \cite{maMovableAntennaEnhanced2024}.
		\item \textbf{PMA}: Applies the proposed MA (PMA) positions of the SCA-based algorithm.
	\end{itemize}
	The diagram of the antenna positions of each benchmark schemes is shown in \Cref{MAPos}.
	the parameters in Algorithm \ref{Algorithm 1} is set as $\xi=0.5$, $\delta=10^{-4}$, $\bm{x}^0 = \bm{x}_{\mathrm{CMA}}$, $L = 241$, $\Delta u =0.1$, and  $\tilde{V}_{\text{th}}=\frac{\var(\bm{x}_{\text{USA}})-V_{\text{low}}}{V_{\text{up}}-V_{\text{low}}}$, where $\var(\bm{x}_{\text{USA}})$ is designed to guarantee the estimation fidelity of the proposed MA array by establishing the USA as a rigorous performance floor for the array aperture.
	
	\begin{figure}[t]
		\vspace{-0.75em}
		\centering
		\includegraphics[width=0.9\linewidth]{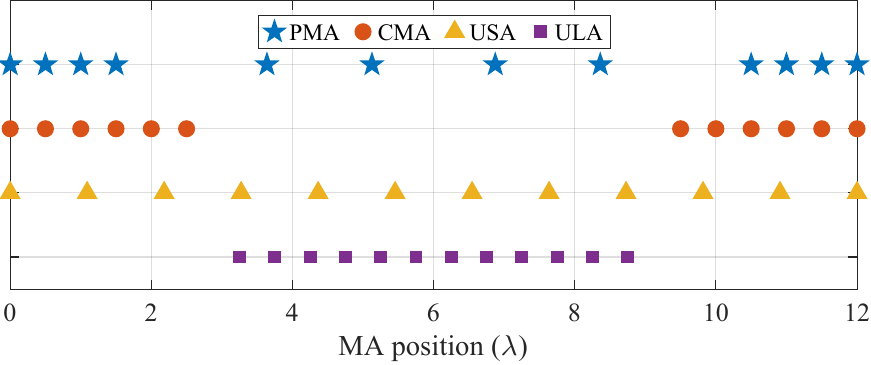}
		\caption{The diagram of antenna positions for benchmark schemes.}
		\label{MAPos}
	\end{figure}
\begin{figure}[t]
	\centering
	\includegraphics[width=0.7\linewidth]{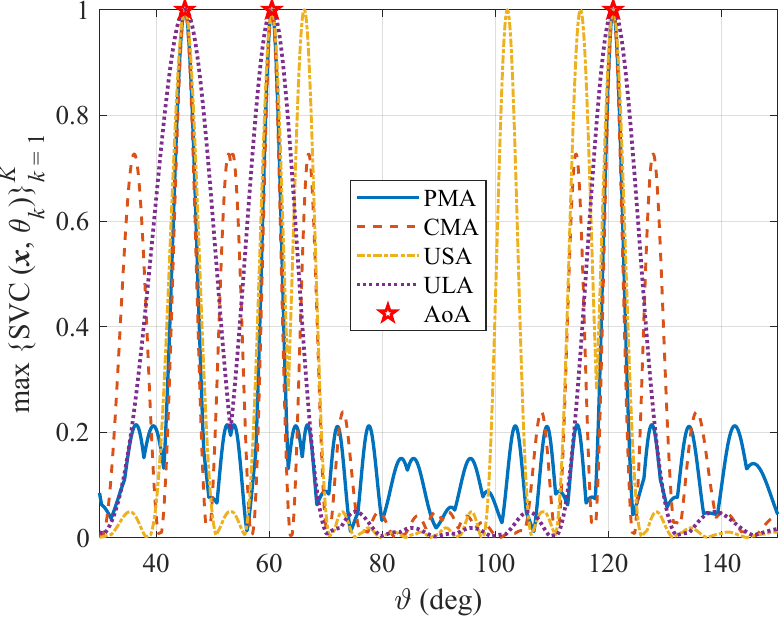}
	\caption{$\mathrm{SVC}$ versus $\vartheta$.}
	\label{beam}
	\vspace{-1em}
\end{figure}

\begin{figure}[t]
	\vspace{-1em}
	\centering
	\includegraphics[width=0.7\linewidth]{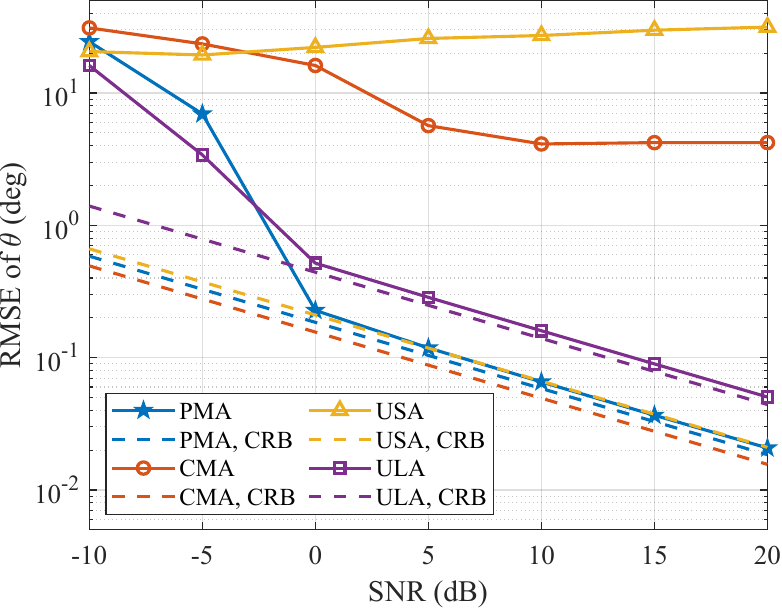}
	\caption{RMSE versus SNR.}
	\label{RMSE-SNR}
	\vspace{-1em}
\end{figure}
\begin{figure}[t]
	\centering
	\includegraphics[width=0.7\linewidth]{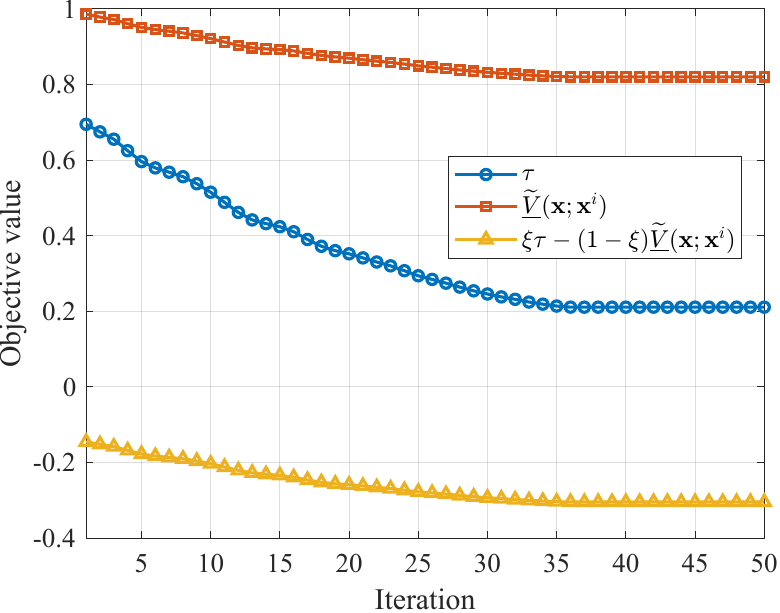}
	\caption{Objective value versus iteration.}
	\label{iter}
	\vspace{-1em}
\end{figure}

			\Cref{beam} calculates the $\mathrm{SVC}$ for each case of MA positions with the AoA range $\vartheta\in \mathcal{A}$. To fully demonstrate the superiority of the proposed algorithm, we plot the AoA that exhibits the maximum $\mathrm{SVC}$ among the $K$ targets for each $\vartheta$, i.e., $\max \{\mathrm{SVC}(\bm{x},\theta_k)\}_{k=1}^K$. One observation is that the proposed MA positions can yield narrower main lobes in the case of multiple targets sensing compared to the benchmark schemes than USA and ULA. Another observation is that compared to MA in \cite{maMovableAntennaEnhanced2024} and USA, the proposed antenna positions exhibit suppressed sidelobes by maintaining the $\mathrm{SVC}$ below $0.1$. This reduction in spatial ambiguity minimizes the probability of misidentifying sidelobes as main lobes, thereby ensuring reliable estimation performance for AoA estimation for multiple targets.
			
			\Cref{RMSE-SNR} compares the RMSE of the AoA estimation versus SNR for different schemes. We can observe that the RMSE of the proposed MA scheme demonstrates remarkable superiority, successfully bypassing the catastrophic spatial ambiguity that plagues both the MA in \cite{maMovableAntennaEnhanced2024} and the USA. By strictly suppressing the SVC in the first optimization stage, the algorithm significantly reduces the threshold SNR, allowing the RMSE to converge to the CRB much earlier than other benchmarks. It is noteworthy that the CRB of the proposed MA exhibits a negligible gap compared to the MA in \cite{maMovableAntennaEnhanced2024}. This marginal difference is due to balancing the suppression of SVC and the maximization of the array aperture. By prioritizing the prevention of spatial ambiguity to ensure global stability, the algorithm achieves superior reliability across a wider SNR range with only a minimal sacrifice in peak estimation precision.
			
			\Cref{iter} illustrates the evolution of the maximum SVC $\tau$ and $\underline{\tilde{V}}(\bm{x})$ with respect to the number of iterations. As the iteration count increases, $\tau$ decreases rapidly. Given that our initialization already places $\underline{\tilde{V}}(\bm{x})$ at its optimum, it undergoes a slow descent as a consequence of the balancing mechanism.  Meanwhile, the total objective function undergoes a decline and subsequently stabilizes at its converged value.
			
			\section{Conclusion}
			In this paper, we investigated the AoA estimation problem for wireless sensing systems utilizing MA arrays. To enhance the estimation performance and accuracy, we formulated a tractable joint optimization framework integrating the SVC and the CRB. By introducing a proxy variable and deriving a generalized CRB lower bound, we developed an SCA-based position optimization algorithm. This algorithm effectively tackles nonconvexities via first order Taylor expansions. Simulation results confirmed that the proposed MA positions significantly suppress sidelobes and exhibits superior AoA estimation performance.

			
			
			%
			\bibliographystyle{IEEEtran}
			\bibliography{mybib_Abbreviation_1}

@misc{ye2026selfcalibrationdoaestimationmovable,
	title={Self-Calibration {DOA} Estimation for Movable Antenna Systems with Antenna Position Errors}, 
	author={Chengzhi Ye and Ruoyu Zhang and Wen Wu and Byonghyo Shim},
	year={2026},
	eprint={2605.23140},
	archivePrefix={arXiv},
	primaryClass={eess.SP},
	url={https://arxiv.org/abs/2605.23140}, 
}

@ARTICLE{ye2026RAsensing,
	author={Ye, Chengzhi and Zhang, Ruoyu and Du, Jincheng and Ma, Wenyan and Wu, Qingqing and Wu, Wen and Zhang, Rui},
	journal={IEEE Wireless Communications Letters},
	title={Rotatable Antenna-Enhanced Wireless Sensing with Uniform Sparse Array via Tensor Decomposition},
	year={2026},
	volume={},
	number={},
	pages={1-5},
	doi={10.1109/LWC.2026.3722674},
	note={doi: 10.1109/LWC.2026.3722674}
}

@ARTICLE{ChengzhiTWCMIMOOFDM,
	author={Ye, Chengzhi and Zhang, Ruoyu and Yao, Lei and Guan, Xinrong and Zhang, Yu and Wu, Wen and Zhang, Rui},
	journal={IEEE Trans. Wireless Commun.}, 
	title={Tensor Decomposition-Based Wireless Sensing for {MIMO}-{OFDM} {ISAC} via Flexible Spatial-Temporal-Spectral Optimization}, 
	year={2026},
}

@ARTICLE{ChengzhiIOTJ20262D,
	author={Ye, Chengzhi and Zhang, Ruoyu and Yao, Lei and Wu, Wen},
	journal={IEEE Internet Things J.}, 
	title={Joint Shape-Position Optimization Enhanced {2D} {DOA} Estimation in Movable Antenna Systems}, 
	year={2026},
	volume={},
	number={},
	pages={1-1},
	doi={10.1109/JIOT.2026.3689968}}

@ARTICLE{11456856,
	author={Chen, Guangyi and Zhang, Ruoyu and Guan, Xinrong and Wu, Qingqing and Ning, Boyu and Zhang, Yu and Wu, Wen and Zhang, Rui},
	journal={IEEE Trans. Wireless Commun.}, 
	title={Movable Antenna-Enabled {MIMO} Integrated Sensing and Communication: A Unified Mutual Information Framework}, 
	year={2026},
	volume={25},
	number={},
	pages={14440-14454},
	doi={10.1109/TWC.2026.3675871}}

@ARTICLE{11218873,
	author={Chen, Guangyi and Zhang, Ruoyu and Guan, Xinrong and Hu, Guojie and Wu, Qingqing and Wu, Wen},
	journal={IEEE Internet Things J.}, 
	title={Energy Efficiency Maximization for Multiuser Communications With Movable Antennas: Joint Beamforming and Antenna Position Design}, 
	year={2026},
	volume={13},
	number={1},
	pages={868-881},
	doi={10.1109/JIOT.2025.3626139}}

@ARTICLE{Chengzhi2025,
	author={Ye, Chengzhi and Zhang, Ruoyu and Hu, Changcheng and Yao, Lei and Wu, Wen and Yuen, Chau},
	journal={IEEE Wireless Commun. Lett.}, 
	title={Robust {DOA} Estimation for Movable Antenna Arrays With Partial Gain and Phase Errors}, 
	year={2026},
	volume={15},
	number={},
	pages={545-549},
	doi={10.1109/LWC.2025.3630221}}

@ARTICLE{mosek2019,
	author = {MOSEK ApS},
	title = {The {MOSEK} optimization toolbox for {MATLAB}},
	year = {2019},
	journal = {User’s Guide and
	Reference Manual}
}

@ARTICLE{DQLISAC2026,
	author={Dai, Qianglong and Zeng, Yong and Wang, Huizhi and You, Changsheng and Zhou, Chao and Cheng, Hongqiang and Xu, Xiaoli and Jin, Shi and Lee Swindlehurst, A. and Eldar, Yonina C. and Schober, Robert and Zhang, Rui and You, Xiaohu},
	journal={IEEE Commun. Surveys Tuts.}, 
	title={A Tutorial on {MIMO}-{OFDM} {ISAC}: From Far-Field to Near-Field}, 
	year={2026},
	volume={28},
	number={},
	pages={4319-4358},
	doi={10.1109/COMST.2025.3650568}}

@INPROCEEDINGS{CanSparse,
	author={Wang, Huizhi and Zeng, Yong},
	booktitle={2023 IEEE Globecom Workshops (GC Wkshps)}, 
	title={Can Sparse Arrays Outperform Collocated Arrays for Future Wireless Communications?}, 
	year={2023},
	volume={},
	number={},
	pages={667-672},
	doi={10.1109/GCWkshps58843.2023.10465094}}

@ARTICLE{SparseMIMOsensing,
	author={Roberts, William and Xu, Luzhou and Li, Jian and Stoica, Petre},
	journal={IEEE Trans. Antennas Propag.}, 
	title={Sparse Antenna Array Design for {MIMO} Active Sensing Applications}, 
	year={2011},
	volume={59},
	number={3},
	pages={846-858},
	doi={10.1109/TAP.2010.2103550}}

@ARTICLE{UltramassiveMIMO,
	author={Faisal, Alice and Sarieddeen, Hadi and Dahrouj, Hayssam and Al-Naffouri, Tareq Y. and Alouini, Mohamed-Slim},
	journal={IEEE Veh. Technol. Mag.}, 
	title={Ultramassive {MIMO} Systems at Terahertz Bands: Prospects and Challenges}, 
	year={2020},
	volume={15},
	number={4},
	pages={33-42},
	doi={10.1109/MVT.2020.3022998}}

@ARTICLE{ruoyu2025Large,
	author={Zhang, Ruoyu and Chen, Guangyi and Cheng, Lei and Guan, Xinrong and Wu, Qingqing and Wu, Wen and Zhang, Rui},
	journal={IEEE Trans. Wireless Commun.}, 
	title={Tensor-Based Channel Estimation for Extremely Large-Scale {MIMO}-{OFDM} With Dynamic Metasurface Antennas}, 
	year={2025},
	volume={24},
	number={7},
	pages={6052-6068},
	doi={10.1109/TWC.2025.3551144}}

@ARTICLE{ISACXUewen,
	author={Luo, Xuewen and Lin, Qingfeng and Zhang, Ruoyu and Chen, Hsiao-Hwa and Wang, Xingwei and Huang, Min},
	journal={IEEE Commun. Surveys Tuts.}, 
	title={{ISAC}—A Survey on Its Layered Architecture, Technologies, Standardizations, Prototypes, and Testbeds}, 
	year={2026},
	volume={28},
	number={},
	pages={485-526},
	doi={10.1109/COMST.2025.3565534}}

@ARTICLE{RuoyuUnified2024,
	author={Zhang, Ruoyu and Cheng, Lei and Wang, Shuai and Lou, Yi and Gao, Yulong and Wu, Wen and Ng, Derrick Wing Kwan},
	journal={IEEE Trans. Wireless Commun.}, 
	title={Integrated Sensing and Communication With Massive {MIMO}: A Unified Tensor Approach for Channel and Target Parameter Estimation}, 
	year={2024},
	volume={23},
	number={8},
	pages={8571-8587},
	doi={10.1109/TWC.2024.3351856}}

@ARTICLE{Jiang6g,
	author={Jiang, Wei and Han, Bin and Habibi, Mohammad Asif and Schotten, Hans Dieter},
	journal={IEEE Open J. Commun. Soc.}, 
	title={The Road Towards {6G}: A Comprehensive Survey}, 
	year={2021},
	volume={2},
	number={},
	pages={334-366},
	doi={10.1109/OJCOMS.2021.3057679}}

@ARTICLE{Saad6G,
	author={Saad, Walid and Bennis, Mehdi and Chen, Mingzhe},
	journal={IEEE Network}, 
	title={A Vision of {6G} Wireless Systems: Applications, Trends, Technologies, and Open Research Problems}, 
	year={2020},
	volume={34},
	number={3},
	pages={134-142},
	doi={10.1109/MNET.001.1900287}}

@article{maMovableAntennaEnhanced2024,
  author={Ma, Wenyan and Zhu, Lipeng and Zhang, Rui},
journal={IEEE Trans. Wireless Commun.}, 
title={Movable Antenna Enhanced Wireless Sensing via Antenna Position Optimization}, 
year={2024},
volume={23},
number={11},
pages={16575-16589},
doi={10.1109/TWC.2024.3443293}}

@article{zhangChannelEstimationMovableantenna2024,
  author={Zhang, Ruoyu and Cheng, Lei and Zhang, Wei and Guan, Xinrong and Cai, Yueming and Wu, Wen and Zhang, Rui},
  journal={IEEE Wireless Commun. Lett.}, 
  title={Channel Estimation for Movable-Antenna {MIMO} Systems via Tensor Decomposition}, 
  year={2024},
  volume={13},
  number={11},
  pages={3089-3093},
  doi={10.1109/LWC.2024.3450592}}

@article{maCompressedSensingBased2023,
author={Ma, Wenyan and Zhu, Lipeng and Zhang, Rui},
journal={IEEE Commun. Lett.}, 
title={Compressed Sensing Based Channel Estimation for Movable Antenna Communications}, 
year={2023},
volume={27},
number={10},
pages={2747-2751},
doi={10.1109/LCOMM.2023.3310535}}

@ARTICLE{BNIngMovableAntennaEnhanced2025,
	author={Ning, Boyu and Yang, Songjie and Wu, Yafei and Wang, Peilan and Mei, Weidong and Yuen, Chau and Bjornson, Emil},
	journal={IEEE Wireless Commun.}, 
	title={Movable Antenna-Enhanced Wireless Communications: General Architectures and Implementation Methods}, 
	year={2025},
	volume={},
	number={},
	pages={1-9},
	doi={10.1109/MWC.013.2400238}}

@ARTICLE{Shaoxiaodan6dsensing,
	author={Shao, Xiaodan and Zhang, Rui and Schober, Robert},
	journal={IEEE Wireless Commun. Lett.}, 
	title={Exploiting Six-Dimensional Movable Antenna for Wireless Sensing}, 
	year={2025},
	volume={14},
	number={2},
	pages={265-269},
	doi={10.1109/LWC.2024.3487966}}

@ARTICLE{9573279,
	author={Zhang, Ruoyu and Shim, Byonghyo and Wu, Wen},
	journal={IEEE Trans. Signal Process.}, 
	title={Direction-of-Arrival Estimation for Large Antenna Arrays With Hybrid Analog and Digital Architectures}, 
	year={2022},
	volume={70},
	number={},
	pages={72-88},
	doi={10.1109/TSP.2021.3119768}}

		\end{document}